\documentclass[12pt]{article}

\usepackage{natbib}
\bibpunct{[}{]}{;}{a}{}{,}
\usepackage{epsfig}
\usepackage{mathpazo}
\usepackage{fullpage}
\usepackage{psfrag}
\usepackage{pifont}
\usepackage{amsmath, amssymb, amsfonts, amscd, xspace}
\usepackage{wrapfig}
\usepackage{tabularx}
\usepackage{color}

\usepackage{amsfonts,color,soul,url}
\usepackage{multirow}

\usepackage{mathrsfs}

\newtheorem{theorem}{Theorem}
\newtheorem{proposition}[theorem]{Proposition}
\newtheorem{lemma}{Lemma}

\newenvironment{proof}[1][Proof]{\begin{trivlist}
  \item[\hskip \labelsep {\bfseries #1}]}{\end{trivlist}}

\newcommand{\V}[1]{\ensuremath{\boldsymbol{#1}}\xspace}

\newcommand{\bea}{\begin{eqnarray}}
\newcommand{\eea}{\end{eqnarray}}
\newcommand{\beaa}{\begin{eqnarray*}}
\newcommand{\eeaa}{\end{eqnarray*}}
\newcommand{\bi}{\begin{itemize}}
\newcommand{\ei}{\end{itemize}}

\begin{document}

\title{Group variable selection via convex Log-Exp-Sum penalty with application to a breast cancer survivor study}
\author{ Zhigeng Geng$^{1}$, Sijian Wang$^{1,2}$, Menggang Yu$^{2}$, Patrick O. Monahan$^4$, \\  Victoria Champion$^5$ and Grace Wahba$^{1,2,3}$ \\
                   \small{$^{1}$Department of Statistics, University of Wisconsin, Madison }\\
                   \small{$^{2}$Department of Biostatistics \& Medical Informatics, University of Wisconsin, Madison} \\
\small{$^{3}$Department of Computer Sciences, University of Wisconsin, Madison }\\                                     
\small{$^{4}$ School of Medicine, Indiana University}\\
\small{$^{5}$ School of Nursing, Indiana University}\\
                   \small{$^{*}$swang@biostat.wisc.edu}
}

\date{June 16, 2013}

\maketitle
\baselineskip=20pt

\begin{abstract}
\noindent{In many scientific and engineering applications, covariates are
naturally grouped.  When the group structures are available among covariates,
people are usually interested in identifying both important groups and important
variables within the selected groups. Among existing successful group variable
selection methods, some methods fail to conduct the within group selection. 
Some methods are able to conduct both group and within group selection, but the
corresponding objective functions are non-convex.  Such a
non-convexity may require extra numerical effort.  In this paper, we propose a
novel Log-Exp-Sum{(LES)} penalty for group variable selection.  The
LES penalty is strictly convex.  It can identify important groups as well as
select important variables within the group.  We develop an efficient
group-level coordinate descent algorithm to fit the model. We also derive
non-asymptotic error bounds 
{and asymptotic group selection consistency}
for our method in the high-dimensional setting where
the number of covariates can be much larger than the sample size. Numerical
results demonstrate the good performance of our method in both variable
selection and prediction.  We applied the proposed method to an American Cancer
Society breast cancer survivor dataset. The findings are clinically meaningful
and lead immediately to testable clinical hypotheses.}

\vspace{.2in}
\noindent
\textbf{Key words}: Breast Cancer Survivor; Group Variable Selection; Lasso; Regularization.

\end{abstract}

\newpage

\section{Introduction}

Breast cancer is the most common cancer in women younger than 45 years of age and is the leading cause of death among females in the United States. 
However, the survival rate for these young women with breast cancer has continuously improved over the past two decades, primarily because of improved therapies. With this long-term survival, it is important to study the quality of life that may be hampered by this traumatic event and by the long-term side effects from related cancer therapies \citep{berry2005effect}.

This paper is motivated by analyzing a dataset from a study funded by the American Cancer Society (ACS), a large quality of life study of breast cancer survivors diagnosed at a young age.  The study included 505 breast cancer survivors (BCS) who were aged 18-45 years old at diagnosis and were surveyed 3-8 years after standard treatments. The study collected many covariates and quality of life outcomes. One outcome that is of particular interest is overall well being (OWB).  It is captured by Campbell's index of well being which is measured from seven questionnaire items \citep{campbell1976quality}. Studying the OWB status after an adversity is of great interest in an increasing body of research to comprehensively understand the negative of a traumatic event, for example, cancer at a young age \citep{zwahlen2010posttraumatic}.

In the present analysis, the covariates include demographic variables and social or behavior construct scores.  These construct scores are quantified by questionnaires that are well studied and documented in the literature \citep{Frank-Stromborg2003}.  Furthermore, these constructs are divided into eight non-overlapping groups: personality, physical health, psychological health, spiritual health, active coping, passive coping, social support and self efficacy.  The constructs in each group are designed to measure the same aspect of the social or behavioral status of a breast cancer survivor from different angles.  In our analysis, we are interested in identifying both important groups and important individual constructs within the selected groups that are related to OWB. These discoveries may help design interventions targeted at these young breast cancer survivors from the perspective of cancer control.  In statistics, this is a group variable selection problem.

Variable selection via penalized likelihood estimation has been an active research area in the past decade. When there is no group structure, many methods have been proposed and their properties have been thoroughly studied, for example, see LASSO \citep{tibshirani1996regression}, SCAD \citep{fan2001variable}, Elastic-Net \citep{zou2005regularization}, COSSO \citep{lin2006component}, SICA \citep{lv2009unified}, MCP \citep{zhang2010nearly}, truncated $L_1$ \citep{shen2011tuncatedl1}, SELO \citep{Dicker2011selo} and references therein. However, when there are grouping structures among covariates, these methods still make selection based
on the strength of individual covariates rather than the strength of the group, and may have inadequate performance. A proper integration of the grouping information into the analysis is hence desired, and that may help boost the signal-to-noise ratio.


Several methods have addressed the group variable selection problem in literature.  \cite{yuan2006model} proposed a group LASSO penalty, which penalizes the $L_2$-norm of the coefficients within each group. \cite{zhao2006grouped} proposed a CAP family of group variable selection penalties.  One specific member in the CAP family is the $L_\infty$-norm penalty, which penalizes the $L_\infty$-norm of the coefficients within each group. These two methods can effectively remove unimportant groups, but a possible limitation with the two methods is that they select variables in an ``all-in-all-out'' fashion, i.e., when one variable in a group is selected, all other variables in the same group are also selected. In our analysis of the ACS dataset, however, we want to keep the flexibility of selecting variables within a group. For example, when a group of constructs is related to OWB, it does not necessarily mean all the individual constructs in this group are related to OWB. We may want to not only remove unimportant groups effectively, but also identify important individual constructs within important groups as well.  To achieve the goal, \cite{huang2009group} and \cite{zhou2010group} independently proposed a group bridge penalty and a hierarchical LASSO penalty, respectively.  These two penalties can do the selection at both group level and within group level. However, one possible drawback of the two methods is that their penalty functions are no longer convex. This non-convexity may cause numerical problems in practical computation, especially when the numbers of groups and covariates are large. \cite{simon2012sparse} proposed a sparse group LASSO penalty, which is a mixture of the LASSO penalty and the group LASSO penalty.  Their objective function is convex, but due to the usage of two penalty functions, the regression coefficients may be over-penalized and have a relatively large bias. In this paper, we propose a new Log-Exp-Sum penalty for group variable selection.  This new penalty is convex, and it can perform variable selection at both group level and within-group level. We propose an effective algorithm based on a modified coordinate descent algorithm \citep{friedman2007pathwise,wu2008coordinate} to fit the model. {The theoretical properties of our proposed method are thoroughly studied.  We establish both the finite sample error bounds and asymptotic group selection consistency of our LES estimator.} The proposed method is applied to the ACS breast cancer survivor dataset.


The paper is organized as follows. In Section 2, we propose the Log-Exp-Sum penalty and present the group-level coordinate descent algorithm. {In Section 3, we develop non-asymptotic inequalities and group selection consistency for our LES estimator} in the high-dimensional setting where the number of covariates is allowed to be much larger than the sample size. Simulation results are presented in Section 4. In Section 5, we apply the proposed method to ACS breast cancer dataset. Finally we conclude the paper with Section 6.

\section{Method}
				
\subsection{Preparation}

We consider the usual regression setup: we have training data, $(\V x_i, y_i),~i=1,\dots,n$, where $\V x_i$ and $y_i$ are a $p$-length vector of covariates and response for the $i$th subject, respectively. We assume the total of $p$ covariates can be divided into $K$ groups.  Let the $k$th group have $p_k$ variables, and we use $\V{x}_{i,(k)} = (x_{i,k1},\dots,x_{i,kp_k})^T$ to denote the $p_k$ covariates in the $k$th group for the $i$th subject. In most of the paper, we assume $\sum_k p_k=p$, i.e., there are no overlaps between groups.  This is also the situation in ACS breast cancer survivor data. We will discuss the situation that groups are overlapped in Section \ref{sec:conclusion}.

To model the association between response and covariates, we consider the following linear regression:
\begin{equation}\label{eq:reg}
y_i =\sum_{k=1}^K\sum_{j=1}^{p_k} x_{i,kj}\beta_{kj} + \epsilon_i,~i=1,\dots,n,
\end{equation}
where $\epsilon_i\stackrel{i.i.d.}{\sim}N(0, \sigma^2)$ are error terms and $\beta_{kj}$'s are regression coefficients. We denote $\V\beta_k=(\beta_{k1},\cdots,\beta_{kp_k})'$ to be the vector of regression coefficients for covariates in the $k$th group. Without loss of generality, we assume the response is centered to have zero mean and each covariate is standardized to have zero mean and unit standard deviation, so the intercept term can be removed from the above regression model.

For the purpose of variable selection, we consider the penalized ordinary least square (OLS) estimation:
\begin{equation}\label{eq:pols}
\min_{\beta_{kj}} \frac{1}{2n}{\sum_{i=1}^n}\Bigl(y_i
-\sum_{k=1}^K\sum_{j=1}^{p_k} x_{i,kj}\beta_{kj} \Bigr)^2+ \lambda J(\V\beta),
\end{equation}
where $J(\V\beta)$ is a certain sparsity-induced penalty function and $\lambda$ is a non-negative tuning parameter.

\cite{yuan2006model} proposed the following group LASSO penalty which is to penalize the $L_2$-norm of the coefficients within each group:
\begin{equation}\label{eq:glasso}
    J(\V\beta)=\sum_{k=1}^K\sqrt{\beta_{k1}^2+\dots+\beta_{kp_k}^2}.
\end{equation}
Due to the singularity of $\|\V\beta_k\|_2$ at $\V\beta_k=\V 0$, some estimated coefficient vector $\hat{\V\beta}_k$ will be exactly zero and hence the corresponding $k$th group is removed from the fitted model.

\cite{zhao2006grouped} proposed penalizing the $L_\infty$-norm of $\V\beta_k$:
\begin{equation}\label{eq:cap}
    J(\V\beta)=\sum_{k=1}^K\max\{\beta_{k1},\dots,\beta_{kp_k}\}.
\end{equation}
The $L_\infty$-norm of $\V\beta_k$ is also singular at $\V\beta_k=0$.  Therefore, some estimated coefficient vector $\hat{\V\beta}_k$ will be exactly zero.

We can see that both $L_2$-norm and $L_\infty$-norm are singular only when the whole vector $\V\beta_k$ is zero. Once a component of $\V\beta_k$ is non-zero, the two norm functions are no longer singular and hence cannot conduct the within group variable selection.

\cite{huang2009group} proposed the following group bridge penalty:
\begin{equation}\label{eq:hlasso}
    J(\V\beta)=\sum_{k=1}^K \Bigl(|\beta_{k1}|+\cdots+|\beta_{kp_k}|\Bigr)^{\gamma},
\end{equation}
where $0<\gamma<1$ is another tuning parameter.

\cite{zhou2010group} independently proposed a hierarchical LASSO penalty.  This penalty decomposes $\beta_{kj}=\gamma_k\theta_{kj}$ and considers
\begin{equation}
    J(\gamma_k, \theta_{kj})=\sum_{k=1}^K |\gamma_k|+\lambda\sum_{k=1}^K\sum_{j=1}^{p_k}|\theta_{kj}|.
\end{equation}

When the groups are not overlapped, the hierarchical LASSO penalty is equivalent to the group bridge penalty with $\gamma=0.5$.  We can see that these two penalties are singular at both $\V\beta_k=\V0$ and $\beta_{kj}=0$ and hence is able to conduct both group selection and within group selection, however, the two objective functions are not convex.

\cite{simon2012sparse} proposed the sparse group LASSO penalty:
\begin{equation}
    J(\V\beta) = s\sum_{k=1}^K\sqrt{\beta_{k1}^2+\dots+\beta_{kp_k}^2}+(1-s)\sum_{k=1}^K\sum_{j=1}^{p_k}|\beta_{kj}|,
\end{equation}
where $0<s<1$ is another tuning parameter.  We can see that, by mixing the LASSO penalty and group LASSO penalty, the sparse group LASSO penalty is convex and is able to conduct both group and within group selection. However, with this penalty, each coefficient is penalized twice, and this may lead to a relatively large bias on the estimation.

\subsection{Log-Exp-Sum penalty}
We propose the following Log-Exp-Sum (LES) penalty:
\begin{equation}\label{eq:les}
J(\V\beta)=\sum_{k=1}^K w_k \log\Big(\exp\{\alpha|\beta_{k1}|\}+\dots+\exp\{\alpha|\beta_{kp_k}|\}\Bigr),
\end{equation}
where $\alpha>0$ is another tuning parameter and $w_k$'s are pre-specified weights to adjust for different group sizes, for example, taking $w_k=p_k/p$. The LES penalty is strictly convex, which can be straightforwardly verified by calculating its second derivative. Similar to other group variable selection penalties, the LES penalty utilizes the group structure and is able to perform group selection.  Meanwhile, the LES penalty is also singular at any $\beta_{kj}=0$ point, and hence is able to conduct the within group selection as well.

The LES penalty has three connections to the LASSO penalty.  First, when each group contains only one variable ($p_k=1$), i.e., there is no group structure, the LES penalty reduces to the LASSO penalty.  Second, when the design matrix $\V X$ is orthonormal, i.e. $\V X'\V X=\V I_p$, the LES penalty {yields} the following penalized estimation:
\begin{equation}
\hat{\beta}_{kj}=sign(\hat{\beta}_{kj}^{ols})\Bigl(|\hat{\beta}_{kj}^{ols}|-n\lambda\alpha w_k \frac{\exp\{\alpha|\hat{\beta}_{kj}|\}}{\sum_{l=1}^{p_k}\exp\{\alpha|\hat{\beta}_{kl}|\}}\Bigr)_+,
\end{equation}
where $\hat\beta_{kj}^{ols}$'s are the unpenalized ordinary least square estimates, and function $(a)_+ \triangleq max\{0,a\}$ for any real number $a$.

Then at the group level, we have the following result:
\begin{equation}\label{2.1}
\sum_{j=1}^{p_k}|\hat\beta_{kj}| = \sum_{j:\hat{\beta}_{kj} \neq 0}|\hat{\beta}_{kj}^{ols}|-n\lambda\alpha w_k\frac{\sum_{j:\hat{\beta}_{kj}\neq 0}\exp\{\alpha|\hat{\beta}_{kj}|\}}{\sum_l\exp\{\alpha|\hat{\beta}_{kl}|\}}
\end{equation}

Note that when each estimate $\hat{\beta}_{kj}$ in the $kth$ group is non-zero, the following equality holds:
\begin{equation}\label{proples8}
    \sum_{j=1}^{p_k}|\hat\beta_{kj}| = \sum_{j=1}^{p_k}|\hat{\beta}_{kj}^{ols}|-n\lambda\alpha w_k \qquad if \quad \hat{\beta}_{kj} \neq 0, \quad \forall \; j =1, \ldots p_k,
\end{equation}
which can be viewed as an extension from the thresholding on the individual coefficient  by the LASSO penalty to the thresholding on the $L_1$-norm of the coefficients in each group.

The third connection between the LES penalty and the LASSO penalty is that, given any design matrix $\V X$ (not necessarily orthonormal) and an arbitrary grouping structure ($p_k\geq1$), the LASSO penalty can be viewed as a limiting case of the LES penalty. To be specific, we have the following proposition.  
\begin{proposition}
Given the data, for any positive number $\gamma$, consider the LASSO estimator and LES estimator as follows:
{
\begin{eqnarray*}
 \hat{\V\beta}^{\mbox{LASSO}} &=& \arg\min_{\V\beta} OLS + \gamma\sum_{k=1}^K\sum_{j=1}^{p_k}|\beta_{kj}| \\
  \hat{\V\beta}^{\mbox{LES}} &=& \arg\min_{\V\beta} OLS + \lambda \sum_{k=1}^K \frac{p_k}{p} \log\Big(\exp\{\alpha|\beta_{k1}|\}+\dots+\exp\{\alpha|\beta_{kp_k}|\}\Bigr),
\end{eqnarray*}}
where $OLS=\frac{1}{2n}\sum_{i=1}^n\Bigl(y_i -\sum_{k=1}^K\sum_{j=1}^{p_k} x_{i,kj}\beta_{kj} \Bigr)^2$.

Then we have
{$$\hat{\V\beta}^{\mbox{LES}}-\hat{\V\beta}^{\mbox{LASSO}}\rightarrow \V{0}, ~\mbox{as}~ \alpha\rightarrow0 ~\mbox{and keeping}~ \lambda\alpha/p=\gamma.$$}
\end{proposition}

The proof is given in the Appendix.

In our ACS Breast Cancer Survivor dataset, the construct scores within the same
group may be highly correlated, because these constructs are designed to measure a same aspect of a subject's social or behavioral status but from different angles.  In this case, we may be interested in
selecting or removing highly correlated construct scores together for an easy
interpretation. The LES penalty has a property that the estimated coefficients
of highly correlated variables within the same group are enforced to be similar
to each other, and hence tends to select or remove these highly correlated
variables together.  To be specific, we have the following proposition.  

\begin{proposition}\label{proposition1}
Let $\hat{\V\beta}$ be the penalized OLS estimation with the LES penalty. If $\hat{\beta}_{ki}\hat{\beta}_{kj}>0$, then we have:
\begin{eqnarray*}
&& |\hat{\beta}_{ki}-\hat{\beta}_{kj}|\leq C\sqrt{2(1-\rho_{ki,kj})},\\
&&{ C=\frac{1}{n\lambda\alpha^2w_k}\sqrt{||\V y||^2_2+2n\lambda\sum_{l=1}^Kw_l\log(p_l)}\exp\Bigl(\frac{1}{2n\lambda w_k}||\V y||^2_2+\sum_{l=1}^K\frac{w_l}{w_k}\log(p_l)\Bigr),}
\end{eqnarray*}
where $\rho_{ki,kj}=X_{ki}^TX_{kj}$ is the sample correlation between $X_{ki}$ and $X_{kj}$ and $C$ is fixed once the data are given and the tuning parameters are specified.
\end{proposition}

The proof is given in the Appendix.

\subsection{Algorithm}\label{sec:alg}

We need to solve the following optimization problem:
\begin{equation}\label{eq:polsles}
\min_{\beta_{kj}} Q(\beta_{kj}) = \frac{1}{2n}\sum_{i=1}^n\Bigl(y_i-\sum_{k=1}^K\sum_{j=1}^{p_k}x_{i,kj}\beta_{kj}\Bigr)^2+\lambda\sum_{k=1}^K w_k\log\Big(\exp\{\alpha|\beta_{k1}|\}+\dots+\exp\{\alpha|\beta_{kp_k}|\}\Big).
\end{equation}

We propose applying the coordinate descent algorithm
\citep{friedman2007pathwise,wu2008coordinate} at the group level.  The key idea
is to find the minimizer of the original high-dimensional optimization problem
(\ref{eq:polsles}) by solving a sequence of low-dimensional optimization
problems, each of which only involves the parameters in one group.  To be
specific, for fixed
$\V{\tilde\beta}_{-k}=(\V{\tilde\beta}_1',\ldots,\V{\tilde\beta}_{k-1}',
\V{\tilde\beta}_{k+1}',\ldots,\V{\tilde\beta}_K)$, define the function
$Q_k=Q(\V\beta_k;\V{\tilde\beta}_{-k})$.  Our group-level coordinate descent
algorithm is implemented by minimizing $Q_k$ with respect to $\V\beta_k$ for
each $k$ at a time,
and using the solution to update $\V\beta$; at the next step, {$Q_{k+1} = Q(\V\beta_{k+1}; \V{\tilde\beta}_{-(k+1)})$} is minimized and the minimizer is again used to update $\V\beta$. In this way, we cycle through the indices $k=1,\ldots,K$. This cycling procedure may repeat for multiple times until some convergence criterion is reached.

{
To minimize $Q_k(\V\beta_k,\V{\tilde\beta}_{-k})$, we need to solve the following optimization problem:
\begin{equation}\label{Algo4}
\min_{\V\beta_k} \frac{1}{2n}||\V{c}-\V{A}\V\beta_k||_2^2 + \lambda w_k\log\Bigl(\exp\{\alpha|\beta_{k1}|\}+\dots+\exp\{\alpha|\beta_{kp_k}|\}\Bigr),
\end{equation}
where $\V c$ is a $n-$length vector with $c_i = y_i - \sum_{l\neq k}\sum_{j=1}^{p_l} x_{i,lj}\tilde\beta_{lj} $, and $\V A$ is a $n\times p_k$ matrix 
with $\V A_{ij}=x_{i,kj}$.} {We propose applying the gradient projection method to get the solution.  This method was shown in \citep{figueiredo2007gradient} to be very computationally efficient. 
To be specific, let $\V{u}$ and $\V{v}$ be two vectors such that ${u}_i=(\beta_{ki})_+$ and ${v}_i=(-\beta_{ki})_+$, 
where $(a)_+ = max\{0,a\}$. Then the optimization problem (\ref{Algo4}) is equivalent to:
\begin{eqnarray}\label{Algo8}
&&{\min_{\V{z}} \frac{1}{2n}\V{z}^T\V{Bz}} + \V{d}^T\V{z} + f(\V{z}) \equiv F(\V{z}) \nonumber\\
&&s.t. \quad \V{z} \geq 0,
\end{eqnarray}
where
\begin{eqnarray}\label{Algo9}
&&\V{z} = \begin{bmatrix} \V{u} \\ \V{v}\end{bmatrix}, \quad \V{d} = \begin{bmatrix} -\V{A}^T\V{c} \\ \V{A}^T\V{c}\end{bmatrix}, \quad
\V{B} = \begin{bmatrix} \V{A}^T\V{A} & -\V{A}^T\V{A} \\ -\V{A}^T\V{A} & \V{A}^T\V{A} \end{bmatrix}, \nonumber\\
&&
{f(\V{z}) = {\lambda
w_k\log\Bigl(\exp\{\alpha(z_1+z_{p_k+1})\}+\dots+\exp\{\alpha(z_{p_k}+z_{2p_k})\}
\Bigr).\nonumber}}
\end{eqnarray}}

Then we apply the Barzilai-Borwein method in \cite{figueiredo2007gradient} to get the solution to (\ref{Algo8}). The details of this method are summarized in the Appendix.

Since our objective function is convex and the LES penalty is separable at the group level, by results in \cite{tseng2001convergence}, our algorithm is guaranteed to converge to the global minimizer. Note that, if we apply the coordinate descent algorithm at the individual coefficient level, the algorithm is not guaranteed to converge, and in our numerical study, we observed that sometimes the updates were trapped in a local area.


\subsection{Tuning parameter selection}\label{sec:tuning}
Tuning parameter selection is an important issue in penalized estimation.  One often proceeds by finding estimators which correspond to a range of tuning parameter values. The preferred estimator is then identified  as the one  for the tuning parameter value to optimize some criterion, such as Cross Validation (CV), Generalized Cross Validation (GCV) \citep{craven1978smoothing}, AIC \citep{akaike1973information}, or BIC \citep{schwarz1978estimating}. It is known that CV, GCV and AIC-based methods favor the model with good prediction performance, while BIC-based method tends to identify the correct model \citep{yang2005can}. To implement GCV, AIC and BIC, one needs to estimate the degrees of freedom ($df$) of an estimated model. For our LES penalty, the estimate of $df$ does not have an analytic form even when the design matrix is orthonormal. Therefore, we propose using the randomized trace method \citep{girard1989fast, hutchinson1989stochastic} to estimate $df$ numerically.

We first briefly review the randomized trace method to estimate $df$ of a model which is linear in response $\V y$. Let $\hat{\V y}$ be the estimation of the response $\V y$ based on the model, which is given by:
\begin{equation}
\label{randTr1}
\hat{\V y}=\V A(\lambda)\V y \triangleq  f(\V y),
\end{equation}
where $\V A(\lambda)$ is the so-called ``influence matrix'', which depends on the design matrix $\V X$ and the tuning parameter $\lambda$, but does not depend on the response $\V y$. \cite{wahba1983bayesian} defined $tr(\V A(\lambda))$, the trace of $\V A(\lambda)$, to be ``equivalent degrees of freedom when $\lambda$ is used''.

The randomized trace method is to estimate $tr(\V A(\lambda))$ based on the fact that, for any random variable $\V\delta$ with mean zero and covariance matrix $\rho^2 I$, we have $tr(\V A(\lambda))=E\V\delta^T\V A(\lambda)\V\delta/\rho^2$. To be specific, we generate a new set of responses by adding some random perturbations to the original observations: $\V y^{new} = \V y + \V\delta$, where $ \V\delta \sim N(\V 0, \rho^2 I)$. Then $tr(\V A(\lambda))$ (and hence $df$) can be estimated by
\begin{eqnarray}\label{Tu7}
{\tilde{df}=\V\delta^T\V A(\lambda)\V\delta/\rho^2 = \V\delta^T\V A(\V y^{new}-\V y)/\rho^2
\approx\frac{\V\delta^T\Bigl(f(\V y+\V\delta)-f(\V y)\Bigr)}{\frac{1}{n}\sum_{i=1}^n\delta_i^2}.}
\end{eqnarray}

To reduce the variance of the estimator $\tilde{df}$, we can first generate $R$ independent noise vectors: $\V\delta_{(r)},~r=1, \ldots, R$, and then estimate $df$ by $\tilde{df}$ using each $\V\delta_{(r)}$.  The final estimate of $df$ can be the average of these $R$ estimates:
\begin{eqnarray}\label{Tu8}
{\hat{df}=\frac{1}{R}\sum_{r=1}^R\frac{\V\delta_{(r)}^T\Bigl(f(\V y+\V\delta_{(r)})-f(\V y)\Bigr)}{\frac{1}{n}\V\delta_{(r)}^T\V\delta_{(r)}},}
\end{eqnarray}
which is called an $R$-fold randomized trace estimate of $df$.

The estimated model by LES is nonlinear in response $\V y$, i.e., in equation (\ref{randTr1}), the influence matrix $\V A$ depends on the response $\V y$.  In general, the estimated models by most penalized estimation methods (like LASSO) are nonlinear in response $\V y$. \cite{lin2000smoothing} and \cite{wahba1995adaptive} discussed using the randomized trace method to estimate $df$ of a nonlinear model.  Following their discussions, when an estimated model is linear, i.e., equation (\ref{randTr1}) is satisfied, we can see that $tr(\V A(\lambda))=\sum_{i=1}^n \frac{\partial f(\V y)_i}{\partial y_i}$. So the randomized trace estimate of $df$ can be viewed as an estimation of $\sum_i\frac{\partial f(\V y)_i}{\partial y_i}$ using divided differences. When the model is nonlinear, a response-independent influence matrix $\V A(\lambda)$ does not exist.
However, the divided differences $\frac{f(\V y+\V\delta)_i-f(\V y)_i}{\delta_i}$ 
generally exist, so we can still estimate $df$ by the routine $R$-fold randomized trace
estimator defined in (\ref{Tu8}).

In our numerical experiments, we found that the 5-fold randomized trace estimator
worked well. Its computation load is no heavier than the 5-fold cross-validation. In
addition, as a proof of concept, we conducted several simulation studies to estimate
$df$ of LASSO. Our simulation results (not presented in the paper) showed that the
5-fold {randomized trace} estimates of $df$ for LASSO were very close to the number
of non-zero estimated regression coefficients, which 
is given in \cite{zou2005regularization} as an 
estimator of $df$ for LASSO. 

\section{Theoretical Results}\label{theo_result}
In this section, we present the theoretical properties of our LES estimator. We are interested in the situation when the number of covariates is much larger than the number of observations, i.e., $p>>n$. {In subsection \ref{sub_nonasym}, we establish the non-asymptotic error bounds for the LES estimator.  In subsection \ref{sub_grpselc}, we study the asymptotic group selection consistency for the LES estimator.} Throughout the whole section, we consider the following LES penalized OLS estimation:
\begin{equation}\label{newAsy1}
\min_{\beta_{kj}}
\frac{1}{2n}\sum_{i=1}^n\Bigl(y_i-\sum_{k=1}^K\sum_{j=1}^{p_k}x_{i,kj}\beta_{kj}
\Bigr)^2 +
\lambda\sum_{k=1}^Kp_k\log\Bigl(\exp\{\alpha|\beta_{k1}|\}+\dots+\exp\{
\alpha|\beta_{kp_k}|\}\Bigr).
\end{equation}

\subsection{Non-asymptotic error bounds}\label{sub_nonasym}
In this subsection, we study the non-asymptotic properties for our LES estimator. Under the situation when $p>>n$,  \cite{bickel2009simultaneous} proved several finite-sample error bounds for the LASSO estimator.  Using a similar idea, we extend the argument in \cite{bickel2009simultaneous} to show that similar finite-sample bounds hold for our LES estimator.



We make the following Restricted Eigenvalue assumption with group
structure (REgroup), which is similar to the Restricted Eigenvalue (RE)
assumption in \cite{bickel2009simultaneous}.

\textbf{REgroup assumption:} Assume group structure is prespecified
and $p$ covariates can be divided into $K$ groups with $p_k$ covariates in each
group. For a positive integer $s$ and any $\V\Delta \in \mathbb{R}^p$, the
following condition holds:
{
\begin{equation}
\kappa(s)\triangleq \min_{G\subseteq\{1,\ldots,K\}, \atop |G|\leq s}
\min_{\forall \V\Delta \neq 0,\atop \sum_{k\notin G}||\V\Delta_k||_1
\leq \sum_{k\in G}(1+2p_k)||\V\Delta_k||_1}
\frac{2||\V X\V\Delta||_2}{\sqrt{n}\sqrt{\sum_{k\in
G}p_k(1+p_k)^2||\V\Delta_k||^2_2}} > 0,    \nonumber
\end{equation}}
where $G$ is a subset of $\{1,\dots,K\}$, and $|G|$ is the cardinality of set
$G$. $\V\Delta_k \in \mathbb{R}^{p_k}$ is a subvector of $\V\Delta$ for the
$k$-th group, i.e. $\V\Delta_k = (\Delta_{k1},\ldots,\Delta_{kp_k})^T$. {We denote $\|\cdot\|_2$ and $\|\cdot\|_1$ to be Euclidean norm and $L_1$-norm, respectively.}

\begin{theorem}\label{newthm:bound}
Consider linear regression model (\ref{eq:reg}). Let $\V\beta^*$ be the vector
of true regression coefficients. Assume the random error terms $\epsilon_1,
\ldots, \epsilon_n$ are i.i.d. from the normal distribution with mean zero and
variance $\sigma^2$. {Suppose the diagonal elements of matrix $\V X^T\V X/n$ are equal
to 1. }

Let $G(\V\beta)$ be the set of indices of groups that contain at least one
nonzero element for a vector $\V\beta$, i.e. $G(\V\beta)=\{ \;k\; | \;\exists\;
j, \;1\leq j \leq p_k,\; s.t : \beta_{kj}\neq 0; 1\leq k \leq K\}$. Assume the
{$REgroup$} assumption holds with $\kappa = \kappa(s)>0$, where
$s=|G(\V\beta^*)|$.  Let $A$ be a real number bigger than $2\sqrt{2}$ and
$\gamma=A\sigma\sqrt{\frac{\log{p}}{n}}$. Let two tuning parameters satisfy
$\lambda\alpha = \gamma$. Denote $\hat{\V\beta}$ to be the solution to
optimization problem (\ref{newAsy1}). Then with probability at least
$1-p^{1-A^2/8}$, the following inequalities hold:
\begin{eqnarray}\label{newAsy2}
\frac{1}{n}||\V X(\hat{\V\beta}-\V\beta^*)||_2^2
&\leq& \frac{16A^2\sigma^2s}{\kappa^2}*\frac{\log{p}}{n}
\end{eqnarray}
{
\begin{eqnarray}\label{newAsy3}
||(\hat{\V\beta}-\V\beta^*)||_1
\leq\frac{16A\sigma s}{\kappa^2}*\sqrt{\frac{\log{p}}{n}}
\end{eqnarray}}
\begin{eqnarray}\label{newAsy5}
||\hat{\V\beta}-\V\beta^*||_2
&\leq& (2\sqrt{s}+1)\frac{8A\sigma\sqrt{s}}{\kappa^2}*\sqrt{\frac{\log{p}}{n}}
\end{eqnarray}
\end{theorem}

The proof is given in the Appendix. From this theorem, under some conditions, for any $n$, $p$ and any design matrix $\V X$, with certain probability, we obtained the upper bounds on the estimation errors with prediction loss, $L_1$-norm loss and Euclidian-norm loss .


We can generalize our non-asymptotic results to asymptotic results if we further assume that $\kappa(s)$ is bounded from zero, i.e., there exists a constant $u>0$, such that $\kappa(s)\geq u>0$, $\forall\,n, \, p$ and $\V X$. Then as $n\rightarrow\infty$, $p\rightarrow\infty$ and $\log p/n\rightarrow0$, we have
\begin{equation}\label{}
      \frac{1}{n}\|\V{X}\hat{\V\beta}-\V{X}{{\V\beta}}^*\|_2^2 \rightarrow0,~\|\hat{\V\beta}-{\V\beta}^*\|_1\rightarrow0,~
  \|\hat{\V\beta}-{\V\beta}^*\|_2\rightarrow0.
\end{equation}
which implies the consistency in estimation, when the tuning parameters are properly selected.


In addition, the LASSO estimator is a special case of our LES estimator when
each group has only one variable. {If in Theorem
\ref{newthm:bound}, $p_k=1\;\forall\, k$, our REgroup assumption is exactly
the same as RE assumption in \cite{bickel2009simultaneous}. And we obtain
exactly the same bounds {in (\ref{newAsy2}) and (\ref{newAsy3})}  as presented in Theorem 7.2 in
\cite{bickel2009simultaneous}.}
Therefore, our results can be viewed as an extension of results in
\cite{bickel2009simultaneous} from the setting without group structure to the
setting with group structure.

{
\subsection{Group selection consistency}\label{sub_grpselc}
In this subsection, we study the asymptotic group selection consistency for the LES estimator when both $p$ and $n$ tend to infinity. We would like to show that, with probability tending to 1, the LES estimator will select all important groups (groups that contain at least one important variable) while remove all unimportant groups. Let $\mathscr{O}$ be the event that there exists a solution $\hat{\V\beta}$ to optimization problem (\ref{newAsy1}) such that $||\hat{\V\beta}_k||_\infty > 0$ for all $k\in G(\V\beta^*)$ and $||\hat{\V\beta}_k||_\infty = 0$ for all $k\notin G(\V\beta^*)$, where $\V\beta^*$ is the vector of true regression coefficients for model (\ref{eq:reg}) and $G(\V\beta^*)$ is the set of indices of groups that contain at least one nonzero element for a vector $\V\beta^*$. We would like to show the group selection consistency as the following:
\begin{equation}\label{gsc1}
P(\mathscr{O})\rightarrow 1, \qquad n\rightarrow \infty
\end{equation}

Following \cite{nardi2008asymptotic}, we make the following assumptions.  For notation simplicity, in the remaining of this subsection, we use $G$ to stand for $G(\V\beta^*)$.
\begin{itemize}
\item[](C1) The diagonal elements of matrix $\V X^T\V X/n$ are equal
to 1; 
\item[](C2) Let $p_0 = \sum_{k\in G}p_k$ and $d_n = \min_{k\in G}||\V\beta_k^*||_\infty$, where $||\cdot||_\infty$ is the $L_\infty$ norm.  Denote $c$ to be the smallest eigenvalue of $\frac{1}{n}\V X_G^T\V X_G$ and assume $c>0$, where $\V X_G$ is the submatrix of $\V X$ formed by the columns whose group index is in $G$.  Assume
\begin{equation}
\frac{1}{d_n}\Bigl[\sqrt{\frac{\log p_0}{nc}} + \frac{\lambda\alpha\sqrt{p_0}}{c}
\max_{k\in G}p_k\Bigr] \rightarrow 0;
\end{equation}
\item[](C3) For some $0<\tau<1$, 
\begin{equation}
||\V X^T_{G^c} \V X^T_G(\V X^T_G \V X_G)^{-1}||_\infty \leq \frac{1-\tau}{\max_{k\in G}p_k},
\end{equation}
where $\V X_{G^c}$ is the submatrix of $\V X$ formed by the columns whose group index is not in $G$;
\item[](C4) 
\begin{equation}
\frac{1}{\lambda\alpha}\sqrt{\frac{\log(p-p_0)}{n}} \rightarrow 0.
\end{equation}
\end{itemize}

Condition $(C1)$ assumes the matrix $\V X$ is standardized, which is a common procedure in practice. Condition $(C2)$ essentially requires the minimum among the strongest signal $||\V\beta^*_k||_\infty$ of important groups can not be too small; and the eigenvalue of $\frac{1}{n}\V X_G^T\V X_G$ does not vanish too fast. Notice that the dimension of $\V X_G$ is $n\times p_0$, this implicitly requires $p_0$ can not be larger than $n$. Condition $(C3)$ controls the multiple regression coefficients of unimportant group covariates $\V X_{G^c}$ on the important group covariates $\V X_{G}$. This condition mimics the irrepresentable condition as in \cite{zhao_model_2006} which assumes no group structure. The bound of condition $(C3)$ actually depends on the choice of weights $w_i$'s. By choosing a different set of weights, we could obtain a more relaxed bound in $(C3)$.
Finally, condition $(C4)$ controls the growth rate of $p$ and $n$ where $p$ can grow at exponential rate of $n$.

\begin{theorem}\label{thm_gsc}
Consider linear regression model (\ref{eq:reg}), under the assumptions $(C1) - (C4)$, the sparsity property (\ref{gsc1}) holds for our LES estimator. 
\end{theorem}

The proofs follow the spirit in \cite{nardi2008asymptotic} and the details are given in Appendix.

}

\section{Simulation studies}
In this section, we perform simulation studies to evaluate the finite sample performance of the LES method, and compare the results with several existing methods, including  LASSO, group LASSO (gLASSO), group bridge/hierarchical LASSO (gBrdige/hLASSO) and sparse group LASSO (sgLASSO). We consider four examples.  All examples are based on the following linear regression model:
$$y_i=\V x_i^T\V\beta^*+\epsilon_i,~ i=1, \ldots, n,
$$
where $\epsilon_i \stackrel{i.i.d.}{\sim} N(0,\sigma^2)$. We chose $\sigma$ to control the signal-to-noise ratio to be 3. The details of the settings are described as follows.

\begin{itemize}
  \item[] \textbf{Example 1 (``All-In-All-Out'')}
  There are $K=5$ groups and $p=25$ variables in total, with 5 variables in each group. We generated $\V x_i\sim N(\V 0,\V I)$. The true $\V\beta^*$ was specified as:
\begin{eqnarray}\label{SS1}
\V\beta^*=[\underbrace{2,2,2,-2,-2}_{group 1},\underbrace{0,0,0,0,0}_{group 2},\underbrace{0,0,0,0,0}_{group 3}, \underbrace{0,0,0,0,0}_{group 4}, \underbrace{0, 0,0,0, 0}_{group 5}]^T
\end{eqnarray}


  \item[] \textbf{Example 2 (``Not-All-In-All-Out'')}
  { There are $K=5$ groups and $p=25$ variables in total, with 5 variables in each group. }We generated $\V x_i\sim N(\V 0,\V\Sigma)$, where $\V\Sigma$ was given by:
\begin{footnotesize}
 \begin{equation*}
  \Sigma =
 \begin{pmatrix}
  P &   &   &   & \\
    & P &   &   & \\
    &   & Q &   & \\
    &   &   & Q & \\
    &   &   &   & Q
 \end{pmatrix}\quad where \quad P=
  \begin{pmatrix}
    1& .7&.7&.1&.1 \\
   .7&  1&.7&.1&.1 \\
   .7& .7& 1&.1&.1 \\
   .1& .1&.1& 1&.7 \\
   .1& .1&.1&.7& 1
 \end{pmatrix} \quad  Q=
  \begin{pmatrix}
    1&.7&.7&.7&.7 \\
   .7& 1&.7&.7&.7 \\
   .7&.7& 1&.7&.7 \\
   .7&.7&.7& 1&.7 \\
   .7&.7&.7&.7& 1
 \end{pmatrix}
 \end{equation*}
 \end{footnotesize}

The true $\V\beta^*$ was specified as:
\begin{eqnarray}\label{SS2}
\V\beta^*=[\underbrace{2,2,2,0,0}_{group 1},\underbrace{2,2,2,0,0}_{group 2},\underbrace{0,0,0,0,0}_{group 3}, \underbrace{0,0,0,0,0}_{group 4}, \underbrace{0, 0,0,0, 0}_{group 5}]^T
\end{eqnarray}

\item[] \textbf{Example 3 (mixture)}
 { There are $K=5$ groups and $p=25$ variables in total, with 5 variables in each group. }We generated $\V x_i\sim N(\V 0,\V\Sigma)$, where $\V\Sigma$ was the same as in simulation
setting $2$. The true $\V\beta^*$ was specified as:
\begin{eqnarray}\label{SS3}
\V\beta^*=[\underbrace{0,0,0,2,2}_{group 1},\underbrace{0,0,0,2,2}_{group 2},\underbrace{1,1,1,1,1}_{group 3}, \underbrace{1,1,1,1,1}_{group 4}, \underbrace{0, 0,0,0, 0}_{group 5}]^T
\end{eqnarray}

  \item[] \textbf{Example 4 (mixture)}
{ There are $K=6$ groups and $p=50$ variables in total. For group $1$, $2$, $4$ and $5$, each contains $10$ variables; for group $3$ and $6$, each contains $5$ variables.} We generated $\V x_i\sim N(\V 0,\V\Phi)$, where $\V\Phi$ is a block diagonal matrix given by $\text{diag}(\V\Sigma,\V\Sigma)$,
$\V\Sigma$ here was the same as in simulation setting $2$ and $3$. The true
$\V\beta^*$ was specified as:
\begin{eqnarray}\label{SS4}
\V\beta^*=[\underbrace{0,0,0,2,2,0,0,0,2,2}_{group 1},\underbrace{1,1,1,1,1,0,0,0,0,0}_{group 2},\underbrace{1,1,1,1,1}_{group 3}, \nonumber\\
\underbrace{0,0,0,0,0,0,0,0,0,0}_{group 4}, \underbrace{0,0,0,0,0,0,0,0,0,0}_{group 5},\underbrace{0,0,0,0,0}_{group 6}]^T
\end{eqnarray}
\end{itemize}


For each setup, the sample size is $n=100$. We repeated simulations 1,000 times.  The LES was fitted using the algorithm described in Section \ref{sec:alg}.  The LASSO was fitted using the R package ``glmnet''.  The group LASSO was fitted using the R package ``grplasso''.  The group bridge/hierarchical LASSO was fitted using the R package ``grpreg''.  The sparse group LASSO was fitted using the R package ``SGL''.

{To select the tuning parameters in each of the five methods, we consider two approaches.  The first approach is based on data validation.  To be specific, in each simulation, besides the training data, we also independently generated a set of tuning data with the same distribution and with a same sample size as the training data.  Then for each tuning parameter, we fitted the model on the training data and used the fitted model to predict the response on the tuning set and calculated the corresponding mean square error (prediction error).  The model with the smallest tuning error was selected.}

Our second approach for tuning parameter selection is BIC, which is defined to be:
\[BIC=\log(||\V y-X\hat{\V\beta}||_2^2/n)+\log n\cdot df/n, \]
where $df$ is the degrees of freedom of an estimated model. This format of BIC
is actually based on the profile likelihood to get rid of $\sigma^2$, the
variance of the errors.  It is used in \cite{wang2007tuning} and was shown to
have a good performance. For the LES method, the $df$ was estimated using the
randomized trace method described in Section \ref{sec:tuning}.  For the LASSO
method, the $df$ was estimated by the number of non-zero estimated coefficients
\citep{zou2005regularization}. For the group LASSO method, following
\cite{yuan2006model}, the $df$ was first derived under the condition when the
design matrix was orthonormal; then the exact same formula was used to
approximate the true $df$ when the design matrix was not orthonormal. For the
group bridge/hierarchical LASSO method, {the $df$ was} estimated as suggested in
\cite{huang2009group}. For the sparse group LASSO method, the corresponding
papers did not talk about how to estimate $df$, and we used the number of
non-zero estimated coefficients as the estimator for its $df$.

To evaluate the variable selection performance of methods, we consider sensitivity (Sens) and
specificity (Spec), which are defined as follows:
\begin{eqnarray*}
  && \mbox{Sens}~=\frac{\#~\mbox{of selected important variables}}{\#~\mbox{of important variables}}\\
  && \mbox{Spec}~=\frac{\#~\mbox{of removed unimportant variables}}{\#~\mbox{of unimportant variables}}.
\end{eqnarray*}
For both sensitivity and specificity, a higher value means a better variable selection performance.

To evaluate the prediction performance of methods, following
\cite{tibshirani1996regression}, we consider the model error (ME) which is
defined as:
$$\mbox{ME}=(\hat{\V\beta}-\V\beta^*)^T\V\Sigma(\hat{\V\beta}-\V\beta^*),$$ where $\hat{\V\beta}$ is the estimated coefficient vector, $\V\beta^*$ is the true coefficient vector, and $\V\Sigma$ is the covariance matrix of the design matrix $\V X$. We would like to acknowledge that the model error is closely related to the predictive mean square error proposed in \cite{wahba1985comparison} and \cite{leng2006note}.

The simulation results are summarized in Table \ref{simulation}.  {In Example $1$, 
the group bridge/hierarchical LASSO method has a very good performance. It produced the 
the lowest model error as well as high specificity.} {This is not surprising,} because 
Example 1 is a relatively simple ``All-In-All-Out'' case, i.e., all covariates in a 
group are either all important or all unimportant.  Under this situation, the non-convex 
penalty in the group bridge/hierarchical LASSO method has an advantage over other methods 
in terms of removing unimportant groups.  Although slightly worse than the group bridge / hierarchical LASSO method, the LES method outperformed other three methods in terms of model error.  Actually, when 
tuning set was used for tuning parameter selection, the model error of the LES method was 
almost same as the model error of the group bridge/hierarchical LASSO method.  The 
specificity of the LES method was slightly worse than the LASSO method, but could be either higher 
or lower than the sparse group LASSO method and the group LASSO method, depending on the tuning criteria. 
The sensitivity of all five methods were similar.

In Example 2, the LES method produced the smallest model error when the tuning set was used to select the tuning parameters. When the BIC tuning criterion was used, the LES method had a smaller model error than the LASSO method and the group LASSO method, but had a higher model error than the group bridge/hierarchical LASSO method and the sparse group LASSO method. No method dominated in specificity. In general, the specificity of
the LES method was better than group LASSO, but may be worse than other methods. All five methods had almost identical sensitivities.


In Example $3$, the LES method produced the smallest model errors no matter which tuning criterion was used. The group LASSO method had the highest sensitivity, but its specificity was very low.  This means that the group LASSO method tended to includ a large amount of variables in the model. The LES method was the second best in sensitivity but it also maintained a high specificity.

Example $4$ is similar to Example $3$, but has more covariates and more complex group structure. The conclusion about comparisons is similar to that in Example 3.



\section{American Cancer Society Breast Cancer Survivor Data Analysis}
In this section, we analyze the data from ACS breast cancer study which was conducted at the Indiana University School of Nursing. The participants of the study were survivors of the breast cancer aged $18-45$ years old at diagnosis and were surveyed between $3-8$ years from completion of chemotherapy, surgery, with or without radiation therapy. The purpose of this study is to find out what factors in the psychological, social and behavior domains are important for the OWB of these survivors. Identification of these factors and establishment of their association with OWB may help develop intervention programs to improve the quality of life of breast cancer survivors.

The variables included in our current analysis are 54 social and behavior construct scores and three demographic variables.  The 54 scores are divided into {eight non-overlapping groups: personality}, physical health, psychological health, spiritual health, active coping, passive coping, {social support and self efficacy}. Each group contains up to $15$ different scores. The three demographic variables are: ``age at diagnosis'' (Agediag), ``years of education'' (Yrseduc) and ``How many months were you in initial treatment for breast cancer'' (Bcmths). We treated each demographic variable as an individual group.  There are 6 subjects with missing values in either covariates or response, and we removed them from our analysis.  In summary, we have 499 subjects and 57 covariates in 11 groups in our analysis.

We applied five methods in the data analysis: LASSO, group LASSO, group bridge/hierarchical LASSO, sparse group LASSO and our LES method.  We randomly split the whole dataset into a training set with {sample size} $n=332$ and a test set with {sample size} $n=167$ (the ratio of two sample sizes is about $2:1$). We fitted models on the training set, using two tuning strategies: 
one uses 10-fold CV, the other uses BIC. The BIC tuning procedure for all of 
the five methods is the same as what we described in the simulation studies. 
We then evaluated the prediction performances on the test set. We repeated
the whole procedure beginning with a new random split 100 times.

The upper part of Table \ref{realdata} summarizes the average number of selected groups over 100 replicates, the average number of selected individual variables over 100 replicates and the average mean square errors (MSE) on test sets over 100 replicates for five methods.  We can see that, for all of five methods, the models selected by the 10-fold CV tuning had smaller MSEs (better prediction performance) than the models selected by the BIC tuning.  As the cost of this gain in prediction performance, the models selected by 10-fold CV tuning included more groups and more individual variables than the models selected by BIC tuning.  We can also see that, our LES methods had the smallest MSE among five methods no matter which tuning strategy was used.

The lower part of Table \ref{realdata} summarizes the selection frequency of each group across 100 replicates.  A group is considered to be selected if at least one variable within the group is selected.  Since there are some theoretical works showing that BIC tuning tends to identify the true model \citep{wang2007tuning}, we focus on the selection results with BIC tuning.  We can see that the psychological health group is always selected by all of five methods. For our LES methods, three other groups have very high selection frequency: spiritual health (91 out of 100), active coping (89 out of 100) and self efficacy (99 out of 100).  These three groups are considered to be importantly associated with OWB in literature. Spirituality is a resource regularly used by patients with cancer coping with diagnosis and treatment \citep{gall2005understanding}. \cite{purnell2009religious} reported  that spiritual well-being was significantly associated with quality of life and traumatic stress after controlling for disease and demographic variables. Self-efficacy is the measure of one's own ability to complete tasks and reach goals, which is considered by psychologists to be important for one to build a happy and productive life \citep{parle1997development}. \cite{rottmann2010self} assessed the effect of self-efficacy and reported a strong positive correlation between self-efficacy and quality of life and between self-efficacy and mood. They also suggested that self-efficacy is a valuable target of rehabilitation programs. Coping refers to ``cognitive and behavioral efforts made to master, tolerate, or reduce external and internal demands and conflicts'' \citep{folkman1980analysis}. The coping strategies are usually categorized into two aspects: active coping and passive coping \citep{carrico2006reductions}. Active coping efforts are aimed at facing a problem directly and determining possible viable solutions to reduce the effect of a given stressor. Meanwhile, {passive coping} refers to behaviors that seek to escape the source of distress without confronting it \citep{folkman1985if}. Setting aside the nature of individual patients or specific external conditions, there have been consistent findings that the use of active coping strategies produce more favorable outcomes compared to passive coping strategies, such as less pain as well as depression, and better quality of life  \citep{holmes1990differential}. In summary, our data analysis results are intuitively appealing and lead immediately to testable clinical hypotheses.

\section{Conclusion and Discussion}\label{sec:conclusion}
In this paper, we propose a new convex Log-Exp-Sum penalty for group variable selection. The new method keeps the advantage of group LASSO in terms of effectively removing unimportant groups, and at the same time enjoys the flexibility of removing unimportant variables within identified important groups. We have developed an effective group-level coordinate descent algorithm to fit the model. The theoretical property of our proposed method has been thoroughly studied.  We have established non-asymptotic error bounds {and asymptotic group selection consistency} for our proposed method, in which the number of  variables is allowed to be much larger than the sample size. Numerical results indicate that the proposed method works well in both prediction and variable selection.  We also applied our method to the American Cancer Society breast cancer survivor dataset.  The analysis results are clinically meaningful and have potential impact on intervention to improve the quality of life of breast cancer survivors.

The grouping structure we have considered in this paper does not have overlaps. However, it is not unusual for a variable to be a member of several groups. For example, given some biologically defined gene sets, say pathways, not surprisingly, there will be considerable overlaps among these sets. Our LES penalty can be modified for variable selection when the groups have overlaps. With a little change of notation, the $p$ variables are denoted by $X_1,\dots,X_p$ and their corresponding regression coefficients are
$\beta_1,\dots,\beta_p$. Let $V_k\subseteq \{1,2,\dots,p\}$ be the set of indices of variables in the $k$th group.  We consider the
following optimization problem:
\begin{equation}\label{eq:les2}
{\frac{1}{2n}\sum_{i=1}^n}\Bigl(y_i-\beta_0-\sum_{j=1}^px_{i,j}\beta_{j}\Bigr)^2+\lambda\sum_{k=1}^Kw_k\log\Big(\sum_{j\in V_k}\exp\{\alpha m_j |\beta_j|\}\Big),
\end{equation}
where $w_k,~k=1,\ldots,K$ are weights to adjust for possible different size of each group, say, taking $w_k=p_k/p$, and $m_j,~j=1,\ldots,p$ are weights to adjust for possible different frequency of each variable included in the penalty term, say taking $m_j=1/n_j$, where $n_j$ is the number of groups which include the variable $X_j$. It is easy to see that the objective function (\ref{eq:les2}) reduces to the objective function with the original LES penalty (\ref{eq:les}) when there is no overlap among the $K$ groups.

\section{Appendix}
{
\subsection{Proof of Proposition 1}
\begin{proof}:
By KKT condition, $\hat{\V\beta}^{\mbox{LASSO}}$, the solution of LASSO
satisfies:
\begin{equation}\label{proplasso1}
\frac{1}{n}X_{kj}^T(\V y-\V X\hat{\V\beta}^{\mbox{LASSO}})
=\gamma*sign(\hat{\beta}^{\mbox{LASSO}}_{kj}) \quad if\quad
\hat{\beta}^{\mbox{LASSO}}_{kj} \neq 0
\end{equation}
\begin{equation}\label{proplasso2}
\frac{1}{n}|X_{kj}^T(\V y-\V X\hat{\V\beta}^{\mbox{LASSO}})|
\leq  \gamma \quad if\quad \hat{\beta}^{\mbox{LASSO}}_{kj} = 0
\end{equation}

Similarly, $\hat{\V\beta}^{\mbox{LES}}$, the solution of LES satisfies:
\begin{equation}
\label{proples1}
\frac{1}{n}X_{kj}^T(\V y-\V X\hat{\V\beta}^{\mbox{LES}})
=\frac{\lambda \alpha
p_k}{p}*\frac{\exp\{\alpha|\hat{\beta}^{\mbox{LES}}_{kj}|\}}{\sum_{l=1}^{p_k}
\exp\{\alpha|\hat{\beta}^{\mbox{LES}}_{kl}|\}}*sign(\hat{\beta}^{\mbox{LES}}
_{kj}) \quad if\quad \hat{\beta}^{\mbox{LES}}_{kj} \neq 0
\end{equation}
\begin{equation}
\label{proples2}
\frac{1}{n}|X_{kj}^T(\V y-\V X\hat{\V\beta}^{\mbox{LES}})|
\leq  \frac{\lambda \alpha p_k}{p}*\frac{1}{1 + \sum_{l\neq
j}\exp\{\alpha|\hat{\beta}^{\mbox{LES}}_{kl}|\}} \quad if\quad
\hat{\beta}^{\mbox{LES}}_{kj} = 0
\end{equation}

It is easy to see that if we let tuning parameter $\alpha \rightarrow 0$, each
exponential term in the right hand side of the equations (\ref{proples1}) and
(\ref{proples2}) will be close to 1. If we choose the tuning parameter $\lambda$
such that $\frac{\lambda \alpha}{p}=\gamma$, then KKT condition for LES is
approximately the same as KKT condition for LASSO. Therefore we have
$\hat{\V\beta}^{\mbox{LES}}-\hat{\V\beta}^{\mbox{LASSO}}\rightarrow\V 0$. This
completes the proof.
\end{proof}

\subsection{Proof of Proposition 2}
\begin{proof}:
Let $\hat{\V\beta}$ be the solution of LES. 
By plugging in $\V\beta=0$,
we have
\begin{eqnarray}\label{appt 1}
&&\frac{1}{2n}\sum_{i=1}^n\Bigl(y_i-\sum_{k=1}^K\sum_{j=1}^{p_k}x_{i,kj}\hat{
\beta}_{kj}\Bigr)^2+\lambda\sum_{k=1}^K
w_k\log\Bigl(\exp\{\alpha|\hat{\beta}_{k1}|\}+\dots+\exp\{\alpha|\hat{\beta}_{
kp_k}|\}\Bigr)\nonumber\\
&&\leq \frac{1}{2n}||\V y||^2_2+\lambda\sum_{k=1}^Kw_k\log(p_k)
\end{eqnarray}
From (\ref{appt 1}), we have the following two inequalities:
\begin{eqnarray}\label{appt 2}
\sum_{i=1}^n\Bigl(y_i-\sum_{k=1}^K\sum_{j=1}^{p_k}x_{i,kj}\hat{\beta}_{kj}
\Bigr)^2=||\V y-\V X\hat{\V\beta}||^2_2
\leq ||\V y||^2_2+2n\lambda\sum_{k=1}^Kw_k\log(p_k)
\end{eqnarray}
and
\begin{eqnarray}\label{appt 3}
\exp\{\alpha|\hat{\beta}_{k1}|\}+\dots+\exp\{\alpha|\hat{\beta}_{kp_k}|\} \leq
\exp\{\frac{1}{2n\lambda w_k}||\V y||^2_2+\sum_{j=1}^K\frac{w_j}{w_k}\log(p_j)\}
\end{eqnarray}
So, if $\hat{\beta}_{ki}\hat{\beta}_{kj}>0$, by KKT condition:
\begin{eqnarray}\label{appt 4}
X_{ki}^T(\V y-\V X\hat{\V \beta})=n\lambda\alpha w_k
\frac{\exp\{\alpha|\hat{\beta}_{ki}|\}}{\sum_l\exp\{\alpha|\hat{\beta}_{kl}|\}}
sign(\hat{\beta}_{ki})
\end{eqnarray}
\begin{eqnarray}\label{appt 5}
X_{kj}^T(\V y-\V X\hat{\V \beta})=n\lambda\alpha w_k
\frac{\exp\{\alpha|\hat{\beta}_{kj}|\}}{\sum_l\exp\{\alpha|\hat{\beta}_{kl}|\}}
sign(\hat{\beta}_{kj})
\end{eqnarray}
Without lost of generality, we assume
$\hat{\beta}_{ki}\geq\hat{\beta}_{kj}>0$, and
$sign(\hat{\beta}_{ki})=sign(\hat{\beta}_{kj})=1$, then, by taking the  difference between (\ref{appt 4}) and (\ref{appt 5}), we have:
\begin{eqnarray}\label{appt 6}
(X_{ki}-X_{kj})^T(\V y-\V X\hat{\V \beta})=n\lambda\alpha w_k
\frac{\exp\{\alpha\hat{\beta}_{ki}\}-\exp\{\alpha\hat{\beta}_{kj}\}}{\sum_l\exp\{\alpha|\hat{\beta}_{kl}|\}}
\end{eqnarray}
By the convexity of exponential function $\exp\{x\}$, we have:
\begin{eqnarray}\label{appt 7}
\exp\{\alpha\hat{\beta}_{ki}\}-\exp\{\alpha\hat{\beta}_{kj}\}\geq
\exp\{\alpha\hat{\beta}_{kj}\}*\alpha(\hat{\beta}_{ki}-\hat{\beta}_{kj})\geq
\alpha|\hat{\beta}_{ki}-\hat{\beta}_{kj}|
\end{eqnarray}
By inequality (\ref{appt 6}), we have:
\begin{eqnarray}\label{appt 8}
|\exp\{\alpha\hat{\beta}_{ki}\}-\exp\{\alpha\hat{\beta}_{kj}\}| \leq
\frac{\sum_l\exp\{\alpha|\hat{\beta}_{kl}|\}}{n\lambda\alpha
w_k}||X_{ki}-X_{kj}||_2||\V y-\V X\hat{\V\beta}||_2
\end{eqnarray}
Combining (\ref{appt 2}),(\ref{appt 3}), (\ref{appt 7}) and (\ref{appt 8}), we
can have the result we want. This completes the proof.
\end{proof}

\subsection{Barzilai-Borwein algorithm to solve optimization problem (\ref{Algo8})}
\begin{itemize}
	\item[] Step 1 (Initialization): Initialize $\V z$ with some reasonable value $\V z^{(0)}$.  Choose $\phi_{min}$ and $\phi_{max}$, the lower bound and upper bound for the line search step length. Choose $\phi \in [\phi_{min}, \phi_{max}]$, which is an initial guess of the step length. Set the backtracking line search parameter $\rho=0.5$, and set iteration number $t=0$. 
	\item[] Step 2 (Backtracking Line Search): Choose $\phi_{(t)}$ to be the first number in the sequence $\phi, \rho\phi, \rho^2\phi, \ldots$, such that the following ``Armijo condition" is satisfied:
    \begin{eqnarray}
     F\Bigl(\big(\V z^{(t)} - \phi_{(t)}\nabla F(\V z^{(t)})\big)_+\Bigr) \leq F\Bigl(\V z^{(t)}\Bigr), \nonumber
    \end{eqnarray}
    \item[] Step 3 (Update $\V z$) Let $\V z^{(t+1)}\leftarrow \big(\V z^{(t)} - \phi_{(t)}\nabla F(\V z^{(t)})\big)_+$, and compute $\V\delta^{(t)}=\V z^{(t+1)}-\V z^{(t)}$.
    \item[] Step 4 (Update Step Length): Compute $\V\gamma^{(t)}=\nabla F(\V z^{(t+1)})-\nabla F(\V z^{(t)})$, and update step length:
     \begin{eqnarray}
     \phi = median \{\phi_{min}, \frac{||\V\delta^{(t)}||^2_2}{(\V\gamma^{(t)})^T\V\delta^{(t)}}, \phi_{max}\}\nonumber
    \end{eqnarray}
    \item[] Step 5 (Termination) Terminate the iteration if $||\V z^{(t+1)}-\V z^{(t)}||_2$ is small; otherwise, set $t\leftarrow t+1$, and go to step 2. In our numerical experiment, we terminate the iteration with tolerance of 1e-6.
\end{itemize}

\subsection{Proof of Theorem $\ref{newthm:bound}$}

We first prove a useful lemma.

\begin{lemma}
Consider the model (\ref{eq:reg}).  Let $\V\beta^*$ be the true
coefficients of the linear model. Assume the random noise $\epsilon_1, \ldots,
\epsilon_n$ are iid from normal distribution with mean zero and variance
$\sigma^2$. Suppose the diagonal elements of matrix $\V X^T\V X/n$ are equal to 1. Let
$A$ be a real number bigger than $2\sqrt{2}$ and
$\gamma=A\sigma\sqrt{\frac{\log{p}}{n}}$. Let two tuning parameters $\lambda$ and $\alpha$ satisfy
$\lambda\alpha = \gamma$. For any solution $\hat{\V\beta}$ to the minimization
problem (\ref{newAsy1}), and any $\V\beta \in \mathbb{R}^p$, with probability at
least $1-p^{1-A^2/8}$, the following inequality holds:
\begin{eqnarray}\label{newapp1}
&&\frac{1}{n}||\V X(\hat{\V\beta}-\V\beta^*)||^2_2 +
\gamma||\hat{\V\beta}-\V\beta||_{1}\nonumber\\
&\leq&\frac{1}{n}||\V X(\V\beta-\V\beta^*)||^2_2 + 2\gamma\sum_{k\in
G(\V\beta)}(1+p_k)||\hat{\V\beta}_k-\V\beta_k||_{1}
\end{eqnarray}
\end{lemma}

\begin{proof}: Let $\hat{\V\beta}$ be the solution to (\ref{newAsy1}),
then, for any $\V\beta \in
\mathbb{R}^p$, we have:
\begin{equation}\label{newapp3}
\frac{1}{n}||\V y-\V X\hat{\V\beta}||^2_2 + J_{\lambda,\alpha}(\hat{\V\beta})
\leq \frac{1}{n}||\V y-\V X\V\beta||^2_2 + J_{\lambda,\alpha}(\V\beta)
\end{equation}
Here for notation simplicity, we denote $J_{\lambda,\alpha}(\V\beta)\triangleq
2\lambda\sum_{k=1}^Kp_k\log(\exp\{\alpha|\beta_{k1}|\}+\dots+\exp\{\alpha|\beta_
{kp_k}|\})$.

Because $\V y=\V X\V\beta^* + \V\epsilon$, the above inequality is equivalent to:
\begin{equation}\label{newapp4}
\frac{1}{n}||\V X(\hat{\V\beta}-\V\beta^*)||^2_2
\leq \frac{1}{n}||\V X(\V\beta-\V\beta^*)||^2_2 +
\frac{2}{n}\V\epsilon^T\V X(\hat{\V\beta}-\V\beta)+J_{\lambda,\alpha}
(\V\beta) - J_{\lambda,\alpha}(\hat{\V\beta})
\end{equation}

Consider the random event set $\mathscr{A} = \{\frac{2}{n}||\V X^T\V\epsilon||_\infty\leq \gamma\}$, where $||\V X^T\V\epsilon||_\infty = \max_{kj}
|\sum_{i=1}^{n}x_{i,kj}\epsilon_i|$. Because the diagonal elements of matrix
$\V X^T\V X/n$ are equal to 1, the {random variable}
$z_{kj} \triangleq  \frac{1}{\sigma\sqrt{n}}\sum_{i=1}^{n}x_{i,kj}\epsilon_i$
follows the standard normal distribution, even though the $z_{kj}$ may not be
independent from each other. Let $Z$ be another random variable from the standard
normal distribution. For any $kj$, we have:
$Pr(|\sum_{i=1}^{n}x_{i,kj}\epsilon_i|\geq \frac{\gamma n}{2}) = Pr(|Z|\geq
\frac{\gamma \sqrt{n}}{2\sigma})$, then:
\begin{eqnarray}\label{newapp6}
Pr(\mathscr{A}^c) \leq Pr(\cup_{k=1}^K \cup_{j=1}^{p_k} \{|z_{kj}|\geq
\frac{\gamma \sqrt{n}}{2 \sigma}\})
                  \leq p*Pr(|Z|\geq\frac{\gamma \sqrt{n}}{2 \sigma})
                  \leq p^{1-A^2/8}
\end{eqnarray}

Therefore, on the event set $\mathscr{A}$, with probability at least $1-p^{1-A^2/8}$,
we have:
\begin{eqnarray}\label{newapp7}
&&\frac{1}{n}||\V X(\hat{\V\beta}-\V\beta^*)||^2_2  \nonumber\\
&\leq& \frac{1}{n}||\V X(\V\beta-\V\beta^*)||^2_2 +\frac{2}{n}
||\V X^T\epsilon||_\infty||\hat{\V\beta}-\V\beta||_1+J_{\lambda,\alpha}
(\V\beta) - J_{\lambda,\alpha}(\hat{\V\beta})\nonumber\\
&\leq& \frac{1}{n}||\V X(\V\beta-\V\beta^*)||^2_2 +
\gamma||\hat{\V\beta}-\V\beta||_1+J_{\lambda,\alpha}(\V\beta) -
J_{\lambda,\alpha}(\hat{\V\beta})
\end{eqnarray}

Adding $\gamma||\hat{\V\beta}-\V\beta||_1$ to both sides of (\ref{newapp7}), we have:
\begin{eqnarray}\label{newapp8}
&&\frac{1}{n}||\V X(\hat{\V\beta}-\V\beta^*)||^2_2 +\gamma||\hat{\V\beta}-\V\beta||_1
\nonumber\\
&\leq& \frac{1}{n}||\V X(\V\beta-\V\beta^*)||^2_2 +
2\gamma||\hat{\V\beta}-\V\beta||_1+J_{\lambda,\alpha}(\V\beta) -
J_{\lambda,\alpha}(\hat{\V\beta})\nonumber\\
&=&\frac{1}{n}||\V X(\V\beta-\V\beta^*)||^2_2 +
2\gamma\sum_{k\in G(\V\beta)}||\hat{\V\beta}_k-\V\beta_k||_1\nonumber\\
&&+2\lambda\sum_{k\in
G(\V\beta)}p_k\Bigl(\log(\sum_je^{\alpha|\beta_{kj}|})-\log(\sum_je^{\alpha|\hat
{\beta}_{kj}|})\Bigr)\nonumber\\
&&+2\gamma\sum_{k\notin G(\V\beta)}||\hat{\V\beta}_k||_1+2\lambda\sum_{k\notin
G(\V\beta)}p_k\Bigl(\log(p_k)-\log(\sum_je^{\alpha|\hat{\beta}_{kj}|})\Bigr)
\end{eqnarray}

The last equality uses the fact that $\V\beta_k=\V 0$ {, for $k\notin G(\V\beta)$}.

We next show two simple inequalities. Suppose $a_1,\ldots, a_m$ and
$b_1,\ldots,b_m$ are $2m$ arbitrary real numbers, then the following
inequalities hold:
\begin{eqnarray}\label{newapp9}
\sum_{i=1}^m |a_i| \leq m\log\Bigl((e^{|a_1|}+\ldots+e^{|a_m|})/m\Bigr),
\end{eqnarray}
\begin{eqnarray}\label{newapp10}
\log(e^{|a_1|}+\ldots+e^{|a_m|})-\log(e^{|b_1|}+\ldots+e^{|b_m|})\leq\sum_{i=1}
^m|a_i-b_i|
\end{eqnarray}
The first inequality can be shown by using the arithmetic inequality and the geometric mean inequality; the second inequality can be straightforwardly verified by using the
property of log function and the triangle inequality.

By inequality (\ref{newapp9}), since $\lambda\alpha=\gamma$, we have:
\begin{eqnarray}\label{newapp11}
2\gamma\sum_{k\notin G(\V\beta)}||\hat{\V\beta}_k||_1+2\lambda\sum_{k\notin
G(\V\beta)}p_k\Bigl(\log(p_k)-\log(\sum_je^{\alpha|\hat{\beta}_{kj}|})\Bigr)\leq
0
\end{eqnarray}

By inequality (\ref{newapp10}), we have:
\begin{eqnarray}\label{newapp12}
2\lambda\sum_{k\in
G(\V\beta)}p_k\Bigl(\log(\sum_je^{\alpha|\beta_{kj}|})-\log(\sum_je^{\alpha|\hat
{\beta}_{kj}|})\Bigr)\leq2\lambda\alpha\sum_{k\in
G(\V\beta)}p_k||\V\beta_k-\hat{\V\beta}_k||_1
\end{eqnarray}

Simplify inequality (\ref{newapp8}) using (\ref{newapp11}) and (\ref{newapp12}),
we have proved the lemma.

\end{proof}

Now we prove the theorem.
{
\begin{proof}: 
Let $s=|G(\V\beta^*)|$. In the event set
$\mathscr{A}$, let
$\V\beta=\V\beta^*$ in (\ref{newapp1}), then:
\begin{eqnarray}\label{newapp13}
\frac{1}{n}||\V X(\hat{\V\beta}-\V\beta^*)||^2_2 \leq 2\gamma\sum_{k\in
G(\V\beta^*)}(1+p_k)||\hat{\V\beta}_k-\V\beta_k^*||_{1}\leq
2\gamma\sqrt{s}\sqrt{\sum_{k\in
G(\V\beta^*)}p_k(1+p_k)^2||\hat{\V\beta}_k-\V\beta_k^*||^2_2}
\end{eqnarray}

Using a similar argument, we have:
\begin{eqnarray}\label{newapp23}
\sum_{k\notin G(\V\beta^*)}||\hat{\V\beta}_k-\V\beta^*_k||_1
\leq \sum_{k\in G(\V\beta^*)}(1+2p_k)||\hat{\V\beta}_k-\V\beta^*_k||_1
\end{eqnarray}

If {REgroup assumption holds with $\kappa = \kappa(s)$}, then we have:
\begin{eqnarray}\label{newapp14}
\sqrt{\sum_{k\in G(\V\beta^*)}p_k(1+p_k)^2||\hat{\V\beta}_k-\V\beta_k^*||^2_2}
\leq\frac{2||\V X(\hat{\V\beta}-\V\beta^*)||_2}{\kappa\sqrt{n}}
\end{eqnarray}

By (\ref{newapp13}), (\ref{newapp23}) and (\ref{newapp14}), we have:
\begin{eqnarray}\label{newconclusion1}
\frac{1}{n}||\V X(\hat{\V\beta}-\V\beta^*)||^2_2
&\leq&\frac{16s\gamma^2}{\kappa^2}
\end{eqnarray}

Notice that inequality (\ref{newapp1}) implies:
\begin{eqnarray}\label{newconclusion2}
||\hat{\V\beta}-\V\beta^*||_{1}&\leq& 2\sum_{k\in
G(\V\beta^*)}(1+p_k)||\hat{\V\beta}_k-\V\beta_k^*||_{1}\leq 2
\sqrt{s}\sqrt{\sum_{k\in
G(\V\beta^*)}p_k(1+p_k)^2||\hat{\V\beta}_k-\V\beta_k^*||^2_2}\nonumber\\
&\leq&\frac{4\sqrt{s}||\V X(\hat{\V\beta}-\V\beta^*)||_2}{\kappa\sqrt{n}}\leq
\frac{16s\gamma}{\kappa^2}
\end{eqnarray}


Finally, we work on the bound for $||\hat{\V\beta}-\V\beta^*||^2_2$. For
notation simplicity, we denote $\V\delta=\hat{\V\beta}-\V\beta^*$, and let
$G=G(\V\beta^*)$. When the {REgroup assumption} holds, we have:
\begin{eqnarray}\label{newapp15}
||\V\delta_{G^c}||_2 &\leq& ||\V\delta_{G^c}||_1 \leq \sum_{k\in
G(\V\beta^*)}(1+2p_k)||\hat{\V\beta}_k-\V\beta^*_k||_1
\leq 2\sqrt{s}\sqrt{\sum_{k\in
G(\V\beta^*)}p_k(1+p_k)^2||\hat{\V\beta}_k-\V\beta_k^*||^2_2}\nonumber\\
&\leq&\frac{4\sqrt{s}||\V X\V\delta||_2}{\kappa\sqrt{n}}\leq
\frac{16s\gamma}{\kappa^2}
\end{eqnarray}

Because $p_k\geq 1$, then, by the {REgroup assumption},
\begin{eqnarray}\label{newapp16}
||\V\delta_{G}||_2 \leq \sqrt{\sum_{k\in
G(\V\beta^*)}p_k(1+p_k)^2||\hat{\V\beta}_k-\V\beta_k^*||^2_2}\leq\frac{
2||\V X\V\delta||_2}{\kappa\sqrt{n}}
\leq \frac{8\sqrt{s}\gamma}{\kappa^2}
\end{eqnarray}

Therefore, we have:
\begin{eqnarray}\label{newconclusion4}
||\V\delta||_2 &\leq& ||\V\delta_{G}||_2 + ||\V\delta_{G^c}||_2 
\leq (2\sqrt{s}+1)\frac{8\sqrt{s}\gamma}{\kappa^2}
\end{eqnarray}

The inequalities (\ref{newconclusion1}), (\ref{newconclusion2})
and (\ref{newconclusion4}) complete proof of the theorem.
\end{proof}}

{
\subsection{Proof of Theorem $\ref{thm_gsc}$}
\begin{proof}
Let $g_{kj}$ to be the subgradient for LES penalty with respect to $\beta_{kj}$, then:
\begin{equation}
g_{kj} = \lambda\alpha p_k \frac{\exp\{\alpha|\hat{\beta}_{kj}|\}}{\sum_{l=1}^{p_k}
\exp\{\alpha|\hat{\beta}_{kl}|\}}*\partial|\hat{\beta}_{kj}|,
\end{equation}
where $\partial|\hat{\beta}_{kj}| = sign(\hat{\beta}_{kj})$ if $\hat{\beta}_{kj}\neq 0$; and $\partial|\hat{\beta}_{kj}| \in [-1,1]$ if $\hat{\beta}_{kj}= 0$.

Because the LES penalized OLS estimation (\ref{newAsy1}) is convex, by KKT condition, event $\mathscr{O}$ holds if and only if the following two equations are satisfied:
\begin{equation}\label{gsc2}
\hat{\V\beta}_G = \V\beta_G^* + (\frac{1}{n}\V X^T_G \V X_G)^{-1}(\frac{1}{n}\V X^T_G \V\epsilon - \V g_G),
\end{equation}
\begin{equation}\label{gsc3}
\V g_{G^c} = \frac{1}{n}\V X^T_{G^c} \V\epsilon + \frac{1}{n}\V X^T_{G^c} \V X_G (\frac{1}{n}\V X^T_G \V X_G)^{-1} (\V g_G - \frac{1}{n}\V X^T_G \V\epsilon),
\end{equation}
where $\V g_G$ is a vector of subgradient $g_{kj}$'s with $k\in G$ and $\V g_{G^c}$ is a vector of subgradient $g_{kj}$'s with $k\notin G$. 

In order to prove the theorem, we only need to show the following two limits:
\begin{equation}\label{gsc4}
P(||\hat{\V\beta}_G - \V\beta^*_G||_\infty < d_n) \rightarrow 1, \qquad n\rightarrow\infty;
\end{equation} 
and
\begin{equation}\label{gsc5}
P(||\V g_{G^c}||_\infty < \lambda\alpha) \rightarrow 1, \qquad n\rightarrow\infty.
\end{equation} 

Recall $d_n = \min_{k\in G}||\V\beta_k^*||_\infty$, therefore (\ref{gsc4}) implies $||\hat{\V\beta}_k||_\infty > 0$, for all $k\in G$. When $||\hat{\V\beta}_k||_\infty > 0$, if we let $|\hat{\beta}_{kj}| = ||\hat{\V\beta}_k||_\infty > 0$, then:
\begin{equation}\label{gsc6}
|g_{kj}| = \frac{\lambda \alpha p_k\exp\{\alpha|\hat{\beta}_{kj}|\}}{\exp\{\alpha|\hat{\beta}_{k1}|\} + \ldots + \exp\{\alpha|\hat{\beta}_{kp_k}|\}} = \frac{\lambda \alpha p_k}{1+\sum_{l \neq j }\exp\{\alpha(|\hat{\beta}_{kl}|-|\hat{\beta}_{kj}|)\}}\geq \lambda\alpha
\end{equation}  
By inequality (\ref{gsc6}), we know (\ref{gsc5}) implies $||\hat{\V\beta}_k||_\infty = 0$, for all $k\notin G$. In turn, (\ref{gsc4}) and (\ref{gsc5}) imply 
\begin{equation}
P(\mathscr{O})\rightarrow 1, \qquad n\rightarrow \infty  \nonumber
\end{equation}
as claimed. We will use equations (\ref{gsc2}) and (\ref{gsc3}) to show (\ref{gsc4}) and (\ref{gsc5}).

We will first show (\ref{gsc4}). For simplicity, we denote $\V \Sigma = \frac{1}{n}\V X^T_G \V X_G$. Consider the $p_0$-dimensional vector $\V Z = \frac{1}{n}\V \Sigma^{-1}\V X^T_G \V\epsilon$. Then the expectation and covariance of $\V Z$ are $E(\V Z) = \V 0$ and $Var( \V Z) = \frac{\sigma^2}{n}\V \Sigma^{-1}$. By \cite{ledoux2011probability}, the maximum of a Gaussian vector is bounded by:
\begin{equation}\label{gsc7}
E(||\V Z||_\infty) \leq 3\sigma \sqrt{\frac{\log p_0}{n c}}.
\end{equation}

Notice that:
\begin{eqnarray}\label{gsc8}
||\V \Sigma^{-1}\V g_G||_\infty\leq ||\V \Sigma^{-1}||_\infty ||\V g_G||_\infty\leq \sqrt{p_0} ||\V \Sigma^{-1}||_2*\lambda\alpha p_k\leq \frac{\sqrt{p_0}}{c} \lambda\alpha p_k.
\end{eqnarray}

By (\ref{gsc7}), (\ref{gsc8}) and Markov inequality, we have:
\begin{eqnarray}\label{gsc9}
P(||\hat{\V\beta}_G - \V\beta^*_G||_\infty < d_n) &\leq& \frac{E(||\hat{\V\beta}_G - \V\beta^*_G||_\infty)}{d_n}\nonumber\\
&\leq&\frac{1}{d_n}\Bigl[ E(||\V Z||_\infty) + ||\V \Sigma^{-1}\V g_G||_\infty \Bigr]\nonumber\\
&\leq& \frac{1}{d_n}\Bigl[ 3\sigma \sqrt{\frac{\log p_0}{n c}} + \frac{\sqrt{p_0}}{c} \lambda\alpha \max_{k\in G}p_k \Bigr],
\end{eqnarray}
which goes to zero when $n\rightarrow \infty$ under assumption $(C2)$. Thus (\ref{gsc4}) is established. 

Next we show (\ref{gsc5}). Let
\begin{eqnarray}\label{gsc10}
\V g_{G^c} &=&\frac{1}{n}\V X^T_{G^c} \V X_G \V \Sigma^{-1}\V g_G + \frac{1}{n}\V X^T_{G^c}(I - \V X_G \V \Sigma^{-1} \V X_G^T) \V\epsilon \nonumber\\
&\triangleq& \V u + \V v.
\end{eqnarray}
Then for any $k \in G^c$, $||\V g_{G^c}||_\infty \leq || \V u||_\infty + ||\V v||_\infty$.
By assumption $(C3)$, we have:
\begin{eqnarray}\label{gsc11}
||\V u||_\infty &\leq& ||\V X^T_{G^c} \V X_G(\V X^T_G \V X_G)^{-1}||_\infty ||\V g_G||_\infty\nonumber\\
&\leq&\lambda\alpha\max_{k\in G}p_k*||\V X^T_{G^c} \V X_G(\V X^T_G \V X_G)^{-1}||_\infty\nonumber\\
&\leq& \lambda\alpha(1-\tau).
\end{eqnarray}

Notice that for any element $v_i$ of $\V v$, $v_i = \frac{1}{n}X_{kl}^T(I - \V X_G \V \Sigma^{-1} \V X_G^T )\V\epsilon$ for some column $kl$ of matrix $\V X$, $k\notin G$. Then we have $E(v_i) = 0$ and
\begin{equation}\label{gsc12}
Var(v_i) = \frac{\sigma^2}{n^2}X_{kl}^T\Bigl[ I - \V X_G \V \Sigma^{-1} \V X_G^T \Bigr]X_{kl}\leq \frac{\sigma^2}{n^2}||X_{kl}||_2^2 = \frac{\sigma^2}{n}.
\end{equation}

By a similar argument we used in (\ref{gsc7}), we have:
\begin{equation}\label{gsc13}
E(||\V v||_\infty) \leq 3\sigma \sqrt{\frac{\log (p-p_0)}{n}}.
\end{equation}

By Markov's inequality,
\begin{equation}\label{gsc14}
P(||\V v||_\infty > \frac{\tau}{2}\lambda\alpha)\leq \frac{6\sigma}{\tau\lambda\alpha}\sqrt{\frac{\log (p-p_0)}{n}},
\end{equation}
which goes to zero by assumption $(C4)$.

Combining (\ref{gsc11}) and (\ref{gsc14}), we have:
\begin{equation}\label{gsc15}
P(||\V g_{G^c}||_\infty \geq(1-\frac{\tau}{2}) \lambda\alpha) \rightarrow 0, \quad n\rightarrow\infty,
\end{equation}
which gives (\ref{gsc5}). This completes the proof. 
\end{proof}
}

\newpage
\begin{table}[htb]
\caption{{\small Simulation results over 1000 replicates. ``Sens'' means sensitivity of variable selection; ``Spec'' means specificity of variable selection; 
``ME'' means model error. The numbers in parentheses are the corresponding standard
errors. The bold numbers are significantly better than others at a significance level of $0.05$.}}

\label{simulation}
\begin{center}
{\footnotesize 
\begin{tabular}{l|ccc|ccc}
\hline
& \multicolumn{3}{|c|}{Tuning Set Tuning} & \multicolumn{3}{|c|}{BIC Tuning}\\
\hline
Method    & Sens    & Spec   & ME           		& Sens        & Spec   & ME\\
\hline
\hline
\multicolumn{7}{l}{Simulation 1}\\
\hline
LASSO     & 1.000   & 0.582  &   1.022      		&  1.000      &0.864   & 1.318   \\
          & (0.000) & (0.006)&   (0.016)    		&  (0.000)    &(0.004) & (0.021) \\

gLASSO    & 1.000   & 0.392  &    0.717     		&  1.000      &0.939   & 0.856   \\
          & (0.000) & (0.009)&    (0.012)   		&  (0.000)    &(0.004) &  (0.015) \\

gBrdige
/hLASSO   & 1.000   &\textbf{0.943} & \textbf{0.543}  	&  1.000      &0.908   & \textbf{0.641}   \\
          & (0.000) & (0.004)       & (0.010)         	&  (0.000)    &(0.003) & (0.011) \\

sgLASSO   & 1.000   & 0.388  & 0.811  			&  1.000      &\textbf{0.972} & 1.050   \\
          & (0.000) & (0.009)& (0.013)			&  (0.000)    &(0.003)        & (0.018) \\

LES       & 1.000   & 0.463  & \textbf{0.544}  		&  1.000      &0.718   & 0.770   \\
          & (0.000) & (0.006)& (0.010)         		&  (0.000)    &(0.006) & (0.016) \\
\hline
\hline
\multicolumn{7}{l}{Simulation 2}\\
\hline
LASSO     & 0.994   & 0.721  & 2.289  			&  0.987      &0.935    & 2.682   \\
          & (0.001) & (0.006)&(0.035) 			&  (0.001)    &(0.003)  & (0.045) \\

gLASSO    & 1.000   & 0.154  &3.463           		&  \textbf{1.000}      &0.651    & 4.262   \\
          & (0.000) & (0.007)&(0.046)  			&  (0.000)    &(0.006)  & (0.062) \\

gBrdige
/hLASSO   & 0.996   & \textbf{0.854}  & 2.061     	&  0.995      &0.938    & \textbf{2.321} \\
          & (0.001) & (0.004)         & (0.032)   	&  (0.001)    &(0.002)  & (0.040) \\

sgLASSO   & 1.000   & 0.530  & 2.030  			&  0.997      &0.942  & 2.608   \\
          & (0.000) & (0.008)& (0.031)			&  (0.001)    &(0.003)         & (0.045) \\

LES       & 0.999   & 0.534  & \textbf{1.931}  		&  0.994      &0.831    & 2.629   \\
          & (0.000) & (0.009)& (0.031)         		&  (0.001)    &(0.006)  & (0.047) \\
\hline
\hline

\multicolumn{7}{l}{Simulation 3}\\
\hline
LASSO     & 0.899   & 0.590  & 4.158       		&  0.855     &0.854     & 4.886  \\
          & (0.002) &(0.007) & (0.046)     		&  (0.003)   &(0.005)   & (0.066) \\

gLASSO& \textbf{1.000} & 0.025  & 6.018    		&  \textbf{1.000}   &0.239  & 7.011   \\
          & (0.000)    & (0.003)& (0.063)  		&  (0.000)          &(0.007)& (0.082) \\

gBrdige
/hLASSO   & 0.900   & \textbf{0.701}  & 4.337  		&  0.865   &\textbf{0.911}  & 5.559 \\
          & (0.002) & (0.006)         & (0.048)		&(0.002)   &(0.003)         & (0.068) \\

sgLASSO   & 0.970   & 0.327  & 3.563  			&  0.895   &0.809       & 4.738   \\
          & (0.001) & (0.007)& (0.044)			&  (0.002) &(0.005)     & (0.060) \\

LES       & 0.972   & 0.358  & \textbf{3.295}  		&  0.949   &0.635       & \textbf{4.459}   \\
          & (0.002) & (0.008)& (0.041)         		&  (0.003) &(0.009)     & (0.063) \\
\hline
\hline

\multicolumn{7}{l}{Simulation 4}\\
\hline
LASSO     & 0.873  & 0.757  & 5.539                       &   0.826  & 0.920   & 6.830 \\
          & (0.003)&(0.004) & (0.060)                     &  (0.003)  &(0.003) & (0.091) \\

gLASSO  & \textbf{1.000}  & 0.090  & 9.460                &  \textbf{0.998}    & 0.485   & 11.573  \\
          & (0.000)       & (0.005)& (0.089)              &  (0.001)           & (0.006) & (0.150) \\

gBrdige
/hLASSO   & 0.882   & \textbf{0.862}  & 5.058             &  0.853   &\textbf{0.941}  & {6.151} \\
          & (0.003) & (0.003)         & (0.059)           &  (0.003) &(0.002)         & (0.074) \\

sgLASSO   & 0.977  & 0.555  & 4.830                       & 0.893     & 0.924   &  6.996  \\
          & (0.001)& (0.007)& (0.045)                     &(0.003)    & (0.002) & (0.084) \\

LES       & 0.972  & 0.528  & \textbf{4.640}              & 0.966    & 0.755   &  6.301  \\
          & (0.002)& (0.008)& (0.054)                     &(0.002)   & (0.006) & (0.086) \\

\hline

\end{tabular}}
\end{center}
\end{table}

\newpage
\renewcommand\arraystretch{0.95}
\begin{table}[htb]
\caption{Summary of ACS breast cancer survivor data analysis results. Results are based on $100$ random splits. ``Variable selection'' reports the average number of selected individual variables; ``Group selection'' reports the average number of selected groups and ``MSE'' reports the average mean square errors on test sets. The numbers in parentheses are the corresponding standard errors.}
\label{realdata}
\begin{center}
{\footnotesize 
\begin{tabular}{lcccccc}
\hline
\hline
\multicolumn{7}{c}{Selection Frequency and Mean Square Error}\\
\hline
& \multicolumn{3}{c}{10-fold CV Tuning} & \multicolumn{3}{c}{BIC Tuning}\\
\hline
& \multicolumn{1}{c}{Variable} & \multicolumn{1}{c}{Group}&
\multicolumn{1}{c}{MSE}&\multicolumn{1}{c}{Variable}&
\multicolumn{1}{c}{Group}&\multicolumn{1}{c}{MSE}\\
& \multicolumn{1}{c}{selection} & \multicolumn{1}{c}{selection}&
\multicolumn{1}{c}{ }&\multicolumn{1}{c}{selection}&
\multicolumn{1}{c}{selection}&\multicolumn{1}{c}{ }\\
\hline
LASSO   &23.18   &8.76  & 2.6288     & 8.59   &3.97    & 2.7949\\
        &(0.53) &(0.14)&(0.0286)     &(0.33) &(0.15) &(0.0309)  \\

gLASSO&51.32   & 8.80  & 2.6484     &10.46  & 1.64   & 2.8620 \\
        &( 0.42) &(0.12)&(0.0286)   &(0.64) &(0.09) &(0.0307) \\

gBrdige/hLASSO & 16.24  &  3.34  & 2.6239     & 11.56  & 2.94    & 2.7548\\
        & (0.64) & (0.13)&(0.0293)     & (0.26) & (0.07) &(0.0356) \\

sgLASSO     & 25.83  &  8.58  & 2.6221     & 5.89   & 1.59   & 2.8765\\
        & (0.60)& (0.18)&(0.0265)    &(0.31) & (0.11)&(0.0280) \\

LES     &33.50   &9.58  & 2.6072    & 19.86  & 6.37   & 2.7026 \\
        &(0.74) &(0.11)&(0.0283)      &(0.91) &(0.20) &(0.0298)\\
\hline
\multicolumn{7}{c}{} \\
\hline
\multicolumn{7}{c}{Individual Group Selection Frequency}\\
\hline
\hline
\multicolumn{7}{c}{10-fold CV Tuning}\\
\hline
 Group Name & Agediag & Bcmths & Yrseduc & Personaity & Physical & Psychological
\\
 &  &  &  &  & Health & Health \\
   \hline
  LASSO & 73 & 55 & 19 & 96 & 75 & 100 \\
gLASSO & 82 & 66 & 19 & 94 & 82 & 100\\
  gBrdige/hLASSO & 4 & 2 & 0 & 30 & 2 & 100\\
  sgLASSO & 80 & 56 & 18 & 91 & 71 & 100 \\
  LES & 91 & 71 & 21 & 97 & 89 & 100 \\
   \hline
Group Name & Spiritual  &  Active &  Passive  & Social & Self  \\
 & Health &Coping & Coping & Support & Efficacy \\
  \hline
  LASSO & 100 & 100 & 68 & 90 & 100 \\
gLASSO & 98 & 100 & 46 & 93 & 100 \\
  gBrdige/hLASSO & 35 & 70 & 0 & 3 & 88  \\
  sgLASSO & 97 & 98 & 60 & 89 & 98 \\
  LES & 100 & 100 & 86 & 97 & 100 \\
   \hline
   \hline
\multicolumn{7}{c}{BIC Tuning}\\
   \hline
Group Name & Agediag & Bcmths & Yrseduc & Personaity & Physical & Psychological
\\
 &  &  &  &  & Health & Health \\
  \hline
  LASSO & 5 & 1 & 0 & 32 & 14 & 100 \\
gLASSO & 2 & 1 & 4 & 4 & 9 & 100 \\
  gBrdige/hLASSO & 1 & 0 & 0 & 6 & 0 & 100 \\
  sgLASSO & 1 & 0 & 0 & 3 & 4 & 100 \\
  LES & 28 & 14 & 3 & 75 & 48 & 100  \\
   \hline
Group Name & Spiritual  &  Active &  Passive & Social & Self  \\
 & Health &Coping & Coping &  Support & Efficacy \\
  \hline
  LASSO & 63 & 69 & 2 & 14 & 97 \\
gLASSO & 0 & 4 & 1 & 11 & 28\\
  gBrdige/hLASSO & 29 & 67 & 0 & 0 & 91  \\
  sgLASSO & 5 & 10 & 1 & 3 & 32 \\
  LES & 91 & 89 & 31 & 59 & 99 \\
\hline
\end{tabular}}
\end{center}
\end{table}

\end{document}